\documentclass{article}
\usepackage{ijcai25}

\usepackage{times}
\usepackage{soul}
\usepackage{url}
\usepackage[hidelinks]{hyperref}
\usepackage[utf8]{inputenc}
\usepackage[small]{caption}
\usepackage{graphicx}
\usepackage[T1]{fontenc}

\usepackage{booktabs}
\usepackage{algorithm}
\usepackage{algorithmic}
\usepackage[switch]{lineno}
\usepackage{xcolor}
\usepackage{tcolorbox}
\usepackage{tikz}
\usepackage{enumerate}
\usepackage{todonotes}
\usepackage{amsmath}
\usepackage{amsthm}
\usepackage{amssymb}

\tcbset{size=small,top=1mm,bottom=1mm,middle=1mm}

\newtheorem{theorem}{Theorem}[section]
\newtheorem{corollary}[theorem]{Corollary}
\newtheorem{proposition}[theorem]{Proposition}

\theoremstyle{definition}
\newtheorem{definition}[theorem]{Definition}
\newtheorem{example}[theorem]{Example}

\DeclareMathOperator*{\argmax}{arg\,max}

\newcommand{\mechanism}{\mathcal{M}}

\newcommand{\minsum}{\textsc{Min-Sum}}
\newcommand{\minmax}{\textsc{Min-Max}}
\newcommand{\minmaxdec}{\textsc{Min-Max-Dec}}

\def\NP{\textsf{NP}}

\title{Not in My Backyard! Temporal Voting Over Public Chores}

\author{
Edith Elkind$^1$
\and
Tzeh Yuan Neoh$^{2}$\and
Nicholas Teh$^3$\\
\affiliations
$^1$Northwestern University, USA\\
$^2$Harvard University, USA\\
$^3$University of Oxford, UK\\
\emails
edith.elkind@northwestern.edu,
tzehyuan\_neoh@g.harvard.edu,
nicholas.teh@cs.ox.ac.uk
}

\begin{document}

\maketitle

\begin{abstract}
    We study a temporal voting model where voters have dynamic preferences over a set of public chores---projects that benefit society, but impose individual costs on those affected by their implementation. We investigate the computational complexity of optimizing utilitarian and egalitarian welfare. Our results show that while optimizing the former is computationally straightforward, minimizing the latter is computationally intractable, even in very restricted cases. Nevertheless, we identify several settings where this problem can be solved efficiently, either exactly or by an approximation algorithm. We also examine the effects of enforcing temporal fairness and its impact on social welfare, and analyze the competitive ratio of online algorithms. We then explore the strategic behavior of agents, providing insights into potential malfeasance in such decision-making environments. Finally, we discuss a range of fairness measures and their suitability for our setting.
\end{abstract}

\section{Introduction}
The local government is launching a multi-year initiative to enhance community programs and boost tourism, aiming to create a vibrant and sustainable future for the district. At the beginning of every year, the government will unveil a calendar of planned initiatives for the year ahead, featuring events for each month. These initiatives include popular recurring activities such as food and music festivals, farmers' markets, fireworks displays, and sports tournaments.
As part of this effort, residents are invited to participate in a voting process to voice their preferences on which initiatives should take place in the following year.
However, residents living near proposed event locations often raise concerns about potential negative impacts on their quality of life. These concerns include crowding, noise pollution, traffic congestion, or disruptions to the community’s character. Consequently, some members of the community may oppose specific events, even those deemed beneficial for the district as a whole.

This phenomenon is often referred to as \emph{``Not in My Backyard'' (NIMBY)}, and usually comes in the form of an organized effort by local residents or community groups to oppose certain developments or activities in their neighborhood, often due to perceived adverse effects on their local environment or daily lives. 
NIMBY groups often present a significant challenge to public policy design: while they may be driven by legitimate concerns, they can delay or obstruct important community goals. 
This tension between local concerns and broader societal needs makes such movements a complex and often polarizing issue to address.

Thus, the question that we focus on in this work is: how can we effectively handle such preferences to derive an outcome that is \emph{good} for everyone and/or treats everyone \emph{fairly}?

We approach this problem by viewing it through the lens of \emph{temporal voting} \cite{chandak23,elkind2024temporal,elkind2024verifying,lackner2020perpetual,phillips2025strengthening}. However, instead of voters having approval preferences (indicating which candidates they would like to support) over a set of candidates, voters now express disapproval preferences (indicating which candidates they object to) over a set of candidates.

Disapproval preferences are also relevant in settings where there are too many candidates to consider, and agents only have strong opinions about candidates they do not want chosen (which could be due to proximity concerns as mentioned earlier), but are indifferent among the rest (because this choice does not affect them).
However, results for existing models on temporal voting with ``positively-valued'' candidates (sometimes also known as \emph{public goods} or {issues} \cite{conitzer2017fairpublic,fain2018publicgoods,skowron2022proppublic}) do not automatically transfer to this setting of ``negatively-valued'' candidates (which we call \emph{public chores}\footnote{We feel that the term \emph{public chores} reflects the property of these candidates that while they lead to positive outcomes for the community at large, they often require sacrifices that may not be immediately appreciated by some.}); we discuss this briefly at the end of the next subsection.
This is also the case in the standard \emph{fair allocation of indivisible items} model, whereby the fairness concepts or corresponding results for goods and chores could be vastly different.
In these models, items (goods or chores) are considered to be \emph{privately} allocated to agents, and are not shared among agents.
Considering a model with agents that have disutilities over candidates is also relatively new to temporal voting (apart from \emph{costs} associated with implementing projects in participatory budgeting, as opposed to disutilities here).

\subsection{Our Contributions}
We consider a novel variant of the temporal voting model, where agents express disapproval over candidates.

In Section~\ref{sec:welfare}, we investigate the computational complexity of two welfare objectives: minimizing the sum of agents' disutilities (\minsum{}) and minimizing the maximum disutility (\minmax{}).
We show that, while finding a \minsum{} outcome is easy, the  decision problem associated with \minmax{} is NP-complete even in simple settings.
Nevertheless, we provide several positive parameterized complexity results for the \minmax{} objective.
We also provide an approximation ratio for the corresponding optimization problem.

In Section~\ref{sec:potf}, we analyze the price of enforcing temporal constraints with respect to both objectives; for both cases, we derive tight bounds. 

In Section~\ref{sec:sp}, we shift our focus to the study of strategic manipulation with respect to both welfare objectives. 
We show that while \minsum{} is compatible with \emph{strategyproofness} (i.e., no agent can strictly benefit from misreporting their preferences), this is no longer true for \minmax{}.
Additionally, we show that while \minsum{} is incompatible with a stronger \emph{group-strategyproofness} property, finding a `suitable manipulation' is computationally hard.
We present similar results for \minmax{}, but with respect to strategyproofness.

In Section~\ref{sec:online}, we consider the online setting where information on future projects are unknown, and analyze the competitive ratio of online algorithms.

Finally, in Section~\ref{sec:other}, we discuss other popular fairness notions such as proportionality and equitability, and show that they are either not well-defined in the public chores setting, or come at a very high cost to welfare.

We wish to emphasize that many of our results involve multiplicative upper and lower bounds, addressing challenges such as computational intractability (Section~\ref{sec:welfare}), the need for temporal fairness and its impact on welfare (Section~\ref{sec:potf}), and the lack of information relating to future projects (Section~\ref{sec:online}). 
Importantly, the nature and analysis of these multiplicative bounds differs substantially between the `public goods' setting (analogous to temporal voting) and the `public chores' setting (our focus), and they are not directly transferable.

To illustrate, consider an example with $100$ timesteps where the optimal outcome (with respect to the \minmax{} objective) satisfies all agents in $95$ timesteps. In the public goods' setting, an outcome that satisfies all agents in $50$ timesteps would be considered a $\frac{1}{2}$-approximation. However, in the public chores' setting, this same outcome only provides a $10$-approximation.
This distinction in bounds has significant implications for how agents perceive and engage in collective decision-making. Revisiting our motivating example of NIMBY movements, where agents derive limited personal benefit from approved projects, but experience considerable disutility from disapproved ones, it is more accurate to interpret the outcome above as agents being $10$ times more unhappy, rather than $\frac{1}{2}$ as happy.

\subsection{Related Work}
\paragraph{Sequential Decision-Making}
Works in this area mostly consider positively-valued candidates.
Elkind \emph{et al.}~\shortcite{elkind2024temporalelections} is the most relevant one, which looks into the computational questions associated with welfare maximization, and its compatibility with strategyproofness and proportionality (a popular fairness measure) in the setting with positively-valued candidates.
They also look extensively at the welfare cost of mandating proportionality.
Our work also looks at analogous welfare objectives, and we also consider their compatibility with strategyproofness.
However, on top of the distinctions raised earlier, several other key differences are (i) we consider a more general version of the decision problem for \minmax{}, which imposes constraints involving the sequential nature of timesteps; (ii) for manipulability, we consider computational problems in manipulation and also a stronger version of strategyproofness; (iii) we analyze the online setting; (iv) we discuss why proportionality is not well-defined in our setting and consider equitability as an additional fairness measure.

Another relevant line of work in this area is that of \emph{perpetual voting} \cite{lackner2020perpetual,lackner2023proportionalPV}, which focuses on temporal extensions of traditional multiwinner voting rules.
Bulteau \emph{et al.}~\shortcite{bulteau2021jrperpetual}, Chandak \emph{et al.}~\shortcite{chandak23}, Elkind \emph{et al.}~\shortcite{elkind2024verifying}, and Phillips \emph{et al.}~\shortcite{phillips2025strengthening} build on the temporal voting framework and consider temporal extensions of popular proportional representation (or group fairness) axioms. 
Kahana and Hazon~\shortcite{kahana2023seqcollectiveleximin} study several popular allocation rules (round-robin, maximum Nash welfare, and leximin) with respect to approximately proportional fairness axioms.
Kozachinskiy \emph{et al.}~\shortcite{kozachinskiy2025PVdissatisfaction} study a similar model, but look into the conditions under which there is a sublinear growth of dissatisfaction. 
However, we note that their definition of `dissatisfaction' is inherently different from our `disutility' and studying disapprovals more generally: they consider a model where agents still express approval preferences, and the dissatisfaction of an agent is number of rounds where the project chosen was not approved by the
agent.

Bredereck \emph{et al.}~\shortcite{bredereck2020successivecommittee,bredereck2022committeechange} and 
Zech \emph{et al.}~\shortcite{zech2024multiwinnerchange}
look at sequential committee elections, whereby 
an entire committee (set of candidates) is elected in each round, and impose constraints on the extent a committee can change, whilst ensuring that the candidates retain sufficient support from the electorate. 

We refer the reader to the survey by Elkind \emph{et al.}~\shortcite{elkind2024temporal} for an analysis of other works in this area.

\paragraph{Public Decision-Making} Conitzer \emph{et al.}~\shortcite{conitzer2017fairpublic} study several relaxations of proportionality and its axiomatic guarantees towards individual voters. Fain \emph{et al.}~\shortcite{fain2018publicgoods} considered the notion of an approximate core, whereas Skowron and Górecki~\shortcite{skowron2022proppublic} proposed another variant of proportionality that guarantees fairness to groups of voters. Lackner \emph{et al.}~\shortcite{lackner2023freeriding} studied strategic considerations in the same setting. 
Alouf-Heffetz \emph{et al.}~\shortcite{alouf2022better} consider a model of issue-by-issue voting, which can be viewed as a special case of the temporal voting framework.
Again, all works in this setting look at positively-valued issues/candidates and the focus generally revolves around proportionality, which is ill-defined in the setting with negatively-valued candidates (as we demonstrate towards the end of our paper)

\paragraph{Sequential Fair Division} Another related model is that of sequential fair division \cite{cookson2025temporalfairdivision,choo2025proportionalOFD,elkind2024temporalfairdivision,igarashi2023repeated,neoh2025online}. The key difference is that in this model, a single item (good or chore) is \emph{privately} allocated to an agent at each round (i.e., rivalrous in ownership). This is in contrast to our model, which can be thought of as allocating \emph{public} items (i.e., every agent gets a copy of the item chosen at each round).

\paragraph{Fair Allocation of Indivisible Chores} Chores have also been studied in the field of fair division, but  works in the area typically consider private chores that are rivalrous in ownership \cite{AzizCaIg22chores,bogomolnaia2017competitive,dehghani2018envy,Elkind2024choresbudget} (see also recent survey by Aziz \emph{et al.}~\shortcite{aziz2022survey}). The study of private and public goods is known to be vastly different, and this distinction can also be observed in the context of chores. In fact, even in the case of binary chores (i.e., agents having disutility in $\{0,1\}$ for each chore), the difference is stark. Under binary valuations (which is a very popular model in several fair allocation settings \cite{HalpernPrPs20,SuksompongTe22}), the problem of allocating private chores is known to be easy: popular (approximate) notions of fairness such as EF1 can be trivially obtained. However, when considering public chores, the problem becomes non-trivial even under disapproval preferences.
\nocite{neoh2025strategicchores}

\section{Preliminaries}
For each positive integer $z$, let $[z] := \{1,\dots,z\}$.
Let $N = [n]$ be the set of $n$ \emph{agents}, let $P = \{p_1,\dots,p_m\}$ be the set of $m$ \emph{projects} (or \emph{candidates}), and let $T = [\ell]$ be the set of $\ell$ \emph{timesteps}.
Symmetrically to the literature on multiwinner/temporal voting with approval preferences \cite{lackner2022abc,elkind2024temporal}, we will  assume that voters have \emph{disapproval preferences}. 
For each $i \in N$ and $k \in T$, let $D_{ik} \subseteq P$ denote the \emph{disapproval set} of agent $i$ at timestep $k$, and
let the \emph{disapproval vector} of an agent $i$ be $\mathbf{D}_i = (D_{i1}, \dots, D_{i\ell})$.
An {\em instance} of our problem is a tuple $(N, P, T, ({\mathbf D}_i)_{i\in N})$.

An \emph{outcome} is a vector $\mathbf{o} = (o_1,\dots,o_\ell)$, where $o_k\in P$ for each $k \in T$. Let $\Pi({\mathcal I})$ denote the space of all possible outcomes for an instance $\mathcal I$.
For every $k \in T$ and $\mathbf o\in \Pi({\mathcal I})$, the {\em $k$-truncation} of $\mathbf o$ is the vector $\mathbf{o}^{(k)} = (o_1,\dots,o_k)$.
 
The \emph{disutility} of an agent $i \in N$ from  an outcome
$\mathbf{o}$ is given by $d_i(\mathbf{o}) = |\{k \in T: o_k \in D_{ik}\}|$. We extend this definition to truncated outcomes by writing $d_i(\mathbf{o}^{(k)}) = |\{t \in [k]: o_{t} \in D_{it}\}|$.
A {\em mechanism} maps an instance ${\mathcal I} = (N, P, T, ({\mathbf D}_i)_{i\in N})$ 
to an outcome in $\Pi({\mathcal I})$.

We assume that the reader is familiar with basic notions of classic complexity theory~\cite{papadimitriou_computational_2007} 
and parameterized complexity \cite{flum_parameterized_2006,niedermeier_invitation_2006}.
All omitted proofs can be found in the appendix.

\section{Social Welfare Optimization} \label{sec:welfare}
Two commonly studied welfare objectives in 
collective decision-making are maximizing the sum of agents' utilities (i.e., \emph{utilitarian} welfare) or maximizing the utility of the least happy agent (i.e., \emph{egalitarian} welfare). In our context, 
these objectives translate to, respectively, minimizing the sum of agents' disutilities $\sum_{i \in N} d_i(\mathbf{o})$ (we will refer to this as the \minsum{} objective) or minimizing the maximum disutility $\max_{i \in N} d_i(\mathbf{o})$ (we will refer to this as the \minmax{} objective)\footnote{Another popular welfare objective is Nash welfare, which maximizes the geometric mean of agents' utilities.
However, it is well-known that Nash welfare is ill-defined for negative utilities.}.
Accordingly, given 
an instance $\mathcal{I}$, we refer to outcomes $\mathbf{o} \in \Pi(\mathcal{I})$ that minimize $\sum_{i \in N} d_i(\mathbf{o})$ (resp., $\max_{i \in N} d_i(\mathbf{o})$) as \minsum{} (resp., \minmax{}) outcomes.

We will now discuss the complexity of finding \minsum{} and \minmax{} outcomes, starting with the former. 

\subsection{Minimizing the Sum of Agents' Disutilities}\label{sec:util}
To find a \minsum{} outcome, we can greedily select, at each timestep, a project with the lowest number of disapprovals at that timestep. 
It is easy to observe that this greedy algorithm runs in polynomial time.
However, just like in the case of positively-valued projects, the outcomes of the greedy algorithm may be unfair.
Indeed, consider an instance with $2\kappa+1$ agents $N=\{1,\dots, 2\kappa+1\}$, two projects $P = \{p_1,p_2\}$, and $\ell$ timesteps.
Let $\kappa+1$ agents disapprove of $p_2$ at each timestep, and the remaining $\kappa$ agents disapprove of $p_1$ at each timestep.
Then, the \textsc{Min-Sum} algorithm will select $p_1$ at every timestep, favoring the $\kappa+1$ agents and disadvantaging the other $\kappa$ agents. Arguably, this is not a fair outcome.
This motivates us to explore another welfare objective---\minmax{}---that specifically focuses on fairness.

\subsection{Minimizing the Maximum Agents’ Disutility}\label{sec:egal}
When considering the egalitarian welfare, we aim to take into account the temporal nature of our problem, by imposing constraints on agents' disutilities not just at timestep $\ell$, but also at earlier checkpoints. 

Formally, a {\em set of constraints} for an instance $\mathcal I=(N, P, T, ({\mathbf D}_i)_{i\in N})$ is a set
${\mathbf A}=\{(t_1, \lambda_1), \dots, (t_\tau, \lambda_\tau)\}$, where $t_j\in T$, $\lambda_j\in \{0, \dots, \ell\}$ for each $j\in [\tau]$, and $t_1\le\dots\le t_\tau$, $\lambda_1\le\dots\le \lambda_\tau$.

A pair $({\mathcal I}, {\mathbf A})$ defines a decision problem as follows.

\begin{tcolorbox}[title=\minmaxdec{}]
    \textbf{Input}: A problem instance $(N,P,T, (\mathbf{D}_i)_{i\in N})$ and a set of $\tau$ constraints $\mathbf{A}$.
    \tcblower
    \textbf{Question}: Is there an outcome $\mathbf{o}$ such that for each $(t, \lambda )\in \textbf{A}$ it holds that  $\max_{i \in N} d_i({\mathbf o}^{(t)}) \leq \lambda$?
\end{tcolorbox}
In words, each pair $(t_j,\lambda_j) \in \mathbf{A}$ mandates that at timestep $t_j$, the cumulative disutility of every agent should be at most $\lambda_j$.
Having a single constraint $(\ell, \lambda)$ can be seen as a decision version of \minmax{}.
We note that for utilitarian welfare constraints of this form have no impact on the choice of outcome: if there is a solution that satisfies all constraints, then so does the greedy solution described in Section~\ref{sec:util}.

We will now investigate the complexity of \minmaxdec, both in the worst case, and from a parameterized perspective.
In addition to the natural parameters of our problem, i.e., $n$, $m$ and $\ell$, we will consider a parameter $\gamma=\max_{k\in T, p\in P} |\{i: p\in D_{ik}\}|$, i.e., the maximum number of disapprovals that each project has at any timestep.

As a warm-up, we first show that the special case $\gamma=1$ admits a polynomial-time algorithm.
The constraint $\gamma = 1$ is satisfied in settings where each project is ``very local'' and only affects a single agent (but an agent can still disapprove of multiple projects).

\begin{theorem}\label{thm:gamma=1}
    If $\gamma  = 1$, there is a polynomial-time algorithm for \minmaxdec{}.
\end{theorem}
\begin{proof}[Proof sketch]
    Consider a timestep $k \in T$. If there exists a project that receives no disapprovals at timestep $k$, we should select some such project. 
    Thus, without loss of generality, since $\gamma=1$, we can assume that at each timestep each project is disapproved by exactly one agent. 
    
    Construct a bipartite graph $G = (N \times [\lambda_\tau], T, E)$ with parts $N \times [\lambda_\tau]$ and $T$; for all $(i, \lambda) \in N \times [\lambda_\tau], k \in T$ the graph $G$ contains the edge $\{(i, \lambda), k\}$ if and only if agent $i$ disapproves a project at timestep $k$, and for all $\lambda' < \lambda$ and $k' \geq k$, we have $(k',\lambda') \notin \mathbf{A}$. We claim that the maximum matching in $G$ has cardinality $|T|$ if and only if there is an outcome $\mathbf{o}$ satisfying \minmaxdec{}, and, given a size-$|T|$ matching in $G$, we can transform it into an outcome that satisfies all constraints in $\mathbf A$ in polynomial time;
    the proof is deferred to the appendix.
\end{proof}

Notably, increasing $\gamma$ from $1$ to $2$ 
leads to a hardness result for \minmaxdec{},
even if there are at most two projects per timestep and $\tau=|{\mathbf A}|=1$. Importantly, since our hardness proof works for $\tau=1$, it follows that even finding \minmax{} outcomes is computationally hard.
\begin{theorem}\label{thm: NP-hard}
    \minmaxdec{} is \emph{NP}-complete, even with $m = 2$, $\gamma = 2$, and $\tau = 1$.
\end{theorem}

On the positive side, if both the number of agents and the number of constraints are small, \minmaxdec{} becomes tractable.
This provides a practical solution for small-group voting (e.g., when the voting process can be broken down into smaller districts, so that only the people that can be directly affected by the projects get to vote).
\begin{theorem}\label{thm: FPT}
    \minmaxdec{} is \emph{FPT} with respect to $n + \tau$.
\end{theorem}

If only the number of agents is small, we obtain a weaker tractability result.
\begin{theorem}\label{thm: XP}
    \minmaxdec{} is \emph{XP} with respect to $n$.
\end{theorem}

We also prove a similar tractability result (XP) with respect to the number of timesteps $\ell$. 
This result is applicable in settings where the planning timeline is limited (e.g., short-term policy cycles), even if the number of agents or projects can be substantial.
Moreover, we show that our XP result cannot be strengthened to FPT unless FPT = W[2].
\begin{theorem}\label{thm: W[2]}
    \minmaxdec{} is \emph{XP} and \emph{W[2]}-hard with respect to $\ell$.
\end{theorem}

Having explored our problem from the parameterized complexity perspective, we would like to understand whether it admits an approximation algorithm. However,  \minmaxdec{} is inherently a feasibility problem, so it does not have a natural optimization version.
Nevertheless, we can provide some insights for the case $\tau = 1$, which corresponds to the standard notion of egalitarian welfare.

We first present an inapproximability result, derived as a corollary of the proof of Theorem~\ref{thm: NP-hard}: our reduction from \textsc{3-Occur-3SAT} shows that we cannot hope to obtain an
approximation ratio better than $3/2$ in polynomial time.
\begin{corollary}
    For any $\varepsilon > 0$, if there exists a polynomial-time algorithm for \minmaxdec{} with $\tau = 1$ and approximation ratio $\frac{3}{2} - \varepsilon$, then \emph{P = NP}.
\end{corollary}

Now, let $\ell^+$ be the number of timesteps where every project is disapproved by some agent.
Then, we give an algorithm that achieves a $\min(m, 1 + \frac{n^2}{\ell^+})$ approximation ratio, thereby providing an upper bound.

Notably, there are many settings (including the one described in our introductory example) where the number of projects is small (e.g., the projects have been shortlisted by the local government and put up for voting).
Moreover, if $\ell^+$ is large relative to $n$, the quantity $1 + \frac{n^2}{\ell^+}$ is close to $1$.
 
\begin{theorem} \label{thm:2_approx}
    There exists a $\min(m , 1+\frac{n^2}{\ell^+})$-approximation algorithm for the \minmax{} objective.
\end{theorem}
\begin{proof}
    First, for the timesteps where some project receives no disapprovals, we can always choose some such project. Hence, it suffices to consider the $\ell^+$ timesteps in which each project is disapproved by some agent.
    From now on, we will assume that $\ell=\ell^+$.
    
    We construct a polynomial-size integer program for finding \minmax{} outcomes as follows. 
    For each $p \in P$ and $k \in T$, we define a variable $c_{(p,k)}\in \{0,1\}$: $c_{(p, k)}=1$ if and only if project $p$ is selected at timestep $k$.   
    Our constraints require that (1) for each $k\in T$, at least one project has to be chosen in timestep $k$: 
    $\sum\limits_{p\in P} c_{(p,k)}\geq 1$, and (2) the disutility of each agent $i\in N$ is at most $\eta$:
        $\sum\limits_{k\in T} \sum\limits_{p \in D_{ik}} c_{(p,k)}  \leq \eta$. 
    By relaxing the $0$--$1$ variables $c_{(p,k)}$ to take values in $\mathbb{R}_+$, 
    we obtain the following LP relaxation:
\begin{align*}
\text{minimize} \quad & \eta \tag{P1}\label{approx-general-lp}\\
\text{subject to} \quad 
    & \sum\limits_{p\in P} c_{(p,k)} \geq 1, \quad \text{for all } k \in T, \\
    & \sum\limits_{k \in T} \sum\limits_{p \in D_{ik}} c_{(p,k)} \leq \eta, \quad \text{for all } i \in N, \\
    & c_{(p,k)} \geq 0, \quad \text{for all } p \in P \text{ and } k \in T.
\end{align*}   

Let $((c^*_{(p,k)})_{p\in P, k\in T} , \eta )$ be an optimal solution to P1 that lies at a vertex of the respective polytope.  Construct an outcome $\mathbf{o}$ by selecting, at each timestep $k\in T$, a project $p^* \in \argmax_{p \in P} c^*_{(p,k)}$.  We first show that $\max_{i \in N} d_i(\mathbf{o}) \leq m \cdot \eta$ and then argue that $\max_{i \in N} d_i(\mathbf{o}) \leq \left(1 + \frac{n^2}{\ell}\right) \cdot \eta$. As the value of the optimal integer solution is at least $\eta$, this proves the theorem.

For each $i\in N$, let $T_i=\{k: o_k\in D_{ik}\}$ be the set of timesteps where $i$ disapproves the outcome selected by $\mathbf o$. Consider a timestep 
$k \in T$. Since $\sum_{p \in P} c^*_{(p,k)} \geq 1$, we have $c^*_{(o_k,k)} \geq \frac{1}{m}$. 
Thus, for all $i \in N$ we have $d_i(\mathbf{o}) = |T_i| \leq \sum_{k \in T_i} m \cdot c^*_{(o_k,k)} \leq \sum_{k \in T} m \cdot c^*_{(o_k,k)} \leq m \cdot \eta$. 
Hence, $\max_{i \in N} d_i(\mathbf{o}) \leq m \cdot \eta$.

 Next, we will show that $d_i(\mathbf{o}) \leq \left(1 + \frac{n^2}{\ell}\right) \cdot \eta$ for all $i\in N$, and thus $\max_{i \in N} d_i(\mathbf{o}) \leq \left(1 + \frac{n^2}{\ell}\right) \cdot \eta$. We first note that our linear program \ref{approx-general-lp} has 
 $m\ell + 1$ variables and $n + \ell + m\ell$ constraints. As we consider a vertex solution, there are at least $m\ell+1$ constraints that are tight. Thus, at most $n + \ell$ variables of the form $c_{(p,k)}$ are non-zero. Let $T'$ be the set of timesteps where at least two projects are assigned a positive weight by our solution to the LP. For each $k \in T$ we have
 $\sum_{p\in P} c^*_{(p,k)} \geq 1$, so each $k\in T$ contributes at least one non-zero variable. Thus, there are at most $(n+\ell)-|T|=n$ additional non-zero variables, i.e., $|T'|\le n$. 
 Moreover, $c^*_{(o_k, k)}=1$ for all $k\in T\setminus T'$, so for each $i\in N$ we have
 \begin{align*}
 |\{k\in T\setminus T': o_k\in D_{ik}\}|= 
 \sum_{k\in T\setminus T': o_k\in D_{ik}}c^*_{(o_k, k)}\\
 \le \sum_{k\in T\setminus T'}\sum_{p \in D_{ik}} c^*_{(p,k)}\le \sum_{k\in T}\sum_{p \in D_{ik}} c^*_{(p,k)}\le \eta,
 \end{align*}
 where the last transition follows since 
 $((c^*_{(p,k)})_{p\in P, k\in T} , \eta )$ is feasible for P1.
 Therefore, for each $i \in N$ we have $d_i(\mathbf{o})= |\{k\in T: o_k\in D_{ik}\}|
     = |\{k \in  T': o_k\in D_{ik}\}| + |\{k\in T\setminus T': o_k\in D_{ik}\}|  \leq |T'|+\eta \leq n+\eta$. 
     
     Furthermore, as each project is disapproved by at least one agent at each timestep, we have 
     \begin{align*}
     n\cdot \eta&\ge \sum\limits_{i \in N} \sum\limits_{k \in T} \sum\limits_{p \in D_{ik}} c^*_{(p,k)} = \sum\limits_{k \in T} \sum\limits_{i\in N}\sum\limits_{p \in D_{ik}} c^*_{(p,k)} \\
     &\ge\sum\limits_{k \in T}\sum\limits_{p\in P} c^*_{(p,k)}\geq \ell, 
     \end{align*}
     and hence $n^2/\ell\ge n/\eta$. Thus, for all agents $i \in N$ we have $d_i(\mathbf{o}) \leq \frac{\eta + n}{\eta} \cdot \eta = \left(1 + \frac{n}{\eta}\right)  \cdot \eta \leq  \left(1 + \frac{n^2}{\ell}\right) \cdot \eta$.

       It is easy to see that the described algorithm runs in polynomial time with respect to $n$, $m$, and $\ell$.
\end{proof}

\section{Price of Temporal Fairness} \label{sec:potf}
Next, we consider the impact of imposing temporal constraints $\mathbf{A}$ on the agents' welfare. Our analysis belongs to the line of work on
the \emph{price of fairness}, initiated 
by Bei {\em et al.}~\shortcite{bei2021price}. Similar questions have been considered by a number of authors in the multiwinner voting \cite{brill2024priceable,elkind2022price,lackner2020util} and temporal voting literature \cite{elkind2024temporalelections}. Unsurprisingly, our setting, too, exhibits a fundamental tension between fairness and efficiency.

Our analysis applies to both \minsum{} and \minmax{}, but we only consider constraints of the form $\max_{i\in N} d_i({\mathbf o}^{(t)})\le \lambda$. We note that constraints of the form $\sum_{i\in N} d_i({\mathbf o}^{(t)})\le \lambda$ have no impact on the utilitarian welfare, but may reduce egalitarian welfare. While it may be interesting to investigate the impact of utilitarian constraints on egalitarian welfare, this question does not quite fit the price of fairness framework, so we leave it to future work. 

Given a problem instance $\mathcal{I}$ and a set of constraints $\textbf{A} = \{(t_1, \lambda_1), \dots, (t_\tau, \lambda_\tau)\}$, let $\Pi_{\textbf{A}}(\mathcal{I})$ $\subseteq \Pi({\mathcal I})$ denote the set of outcomes for $\mathcal{I}$ that satisfy constraints in $\textbf{A}$. We say that $\mathbf A$ is {\em feasible} if $\Pi_{\textbf{A}}(\mathcal{I})\neq\varnothing$, i.e., if there is an outcome that satisfies all constraints in $\mathbf A$.

For the objectives \minsum{} and \minmax{}, let $\sum$ and $\max$ be the corresponding \emph{welfare operation} (with respect to agents' disutilities).
Then, for a welfare objective $W \in \{\minsum{}, \minmax{}\}$,  we denote by $W_\text{op}(\mathbf{o})$ the \emph{$W$-value} of an outcome $\mathbf{o}$, i.e., the result of applying the welfare operation corresponding to $W$ to $\mathbf{o}$.

We now define the \emph{price of temporal fairness}, which measures the cost of imposing (feasible) temporal constraints $\mathbf{A}$.
\begin{definition}[Price of Temporal Fairness]
    For an objective $W \in \{\minsum{}, \minmax{}\}$, the \emph{price of temporal fairness with respect to $W$} (PoTF$_W$) is the 
    supremum over all instances $\mathcal I$ and feasible constraints $\mathbf{A}$ of the ratio between the minimum $W$-value of an outcome for $\mathcal I$ that satisfies $\textbf{A}$ and the minimum $W$-value of an outcome for $\mathcal I$:
    \begin{equation*}
        \text{PoTF}_W = \sup_{\mathcal{I},\mathbf{A}:\mathbf{A} \text{ feasible}} \frac{\min_{\mathbf{o} \in \Pi_{\textbf{A}}(\mathcal{I})} W_\text{op}(\mathbf{o})}{\min_{\mathbf{o} \in \Pi(\mathcal{I})} W_\text{op}(\mathbf{o})}.
    \end{equation*}
\end{definition}

 We derive tight bounds for both welfare objectives. For \minsum{}, the price of temporal fairness scales with the number of agents.

\begin{theorem} \label{thm:potf_minsum}
    \emph{PoTF}$_\minsum{}=\Theta(n)$.
\end{theorem}
\begin{proof}
    We first prove the upper bound of $\mathcal{O}(n)$.
    Let $T^+=\{k\in T: \cup_{i\in N}D_{ik}=P\}$; for each timestep $k$ in $T^+$ each project is disapproved by at least one agent, so no matter how we select $o_k$, this will contribute to the disutility of some agent. Hence, for each outcome $\mathbf o$ we have $\sum_{i \in N} d_i(\mathbf{o}) \geq |T^+|$. 
    On the other hand, since $\mathbf A$ is feasible, there is an outcome $\mathbf o$ that satisfies all constraints in $\mathbf A$. We modify this outcome as follows: for each $k\in T\setminus T^+$ the set $P\setminus \cup_{i\in N}D_{ik}$ is not empty, so let $o_k$ be some project in this set. For the modified outcome, the disutility of every agent is at most $|T^+|$, so the total disutility is at most $n\cdot |T^+|$. This completes the proof of the upper bound.

    For the lower bound, consider an instance with two projects $P = \{p_1,p_2\}$ and two timesteps.
    Let $D_{11} = D_{12} = \{p_1\}$ and $D_{i1} = D_{i2} = \{p_2\}$ for all other agents $i \in N \setminus \{1\}$. Let $\mathbf{A}$ 
    contain a single constraint $(2, 1)$, which requires that at the end of timestep $2$
    the disutility of every agent is at most $1$. The only outcomes that satisfy this constraint are $(p_1, p_2)$ and $(p_2, p_1)$; under either outcome the total disutility is $n$. On the other hand, the total disutility of $(p_1, p_1)$ is $2$.
    Thus, $\text{PoTF}_\minsum \geq \frac{n}{2} = \Omega(n)$.
\end{proof}

For \minmax{}, the price of fairness is small if there are few constraints or if the number of agents is small.
\begin{theorem} \label{thm:potf_minmax}
    \emph{PoTF}$_\minmax{}=\Theta(\min(\tau,n))$.
\end{theorem}

Theorem~\ref{thm:potf_minmax} shows that one needs to proceed with caution when imposing fairness constraints, especially if the number of agents is large.

\section{Strategic Manipulation} \label{sec:sp}
In our motivating example, ensuring that agents cannot engage in strategic manipulation is vital for maintaining trust and participation in the voting process. 
To address these concerns, we will now focus on agents' strategic considerations with respect to both welfare objectives.

One popular concept in the social choice literature is
that of \emph{strategyproofness}, which states that no agent should be able to strictly benefit (in this case, decrease their disutility) by misreporting their preferences.
It is formally defined as follows. Note that agent $i$'s disutility function $d_i$ is computed with respect to her (truthful) disapproval vector $\mathbf{D}_i$.

\begin{definition}[Strategyproofness] \label{defn-SP}
    For each $i\in N$, let $\mathcal{D}_{-i}$ denote the list of all disapproval vectors except that of agent $i$: 
    $\mathcal{D}_{-i} = (\mathbf{D}_1,\dots,\mathbf{D}_{i-1},\mathbf{D}_{i+1},\dots,\mathbf{D}_n)$.
    A mechanism $\mechanism$ is \emph{strategyproof} (SP) if for each instance $(N, P, T, ({\mathbf D}_i)_{i\in N})$, each
    agent $i \in N$ and each disapproval vector $\mathbf{D}'_i$
    it holds that 
    $d_i(\mechanism(\mathcal{D}_{-i}, \mathbf{D}_i)) \leq d_i(\mechanism(\mathcal{D}_{-i}, \mathbf{D}'_i))$.
    \end{definition}
    
\subsection{Mechanisms for \minsum{}}
    We first show that the greedy algorithm for obtaining a \minsum{} outcome, which  chooses a project with the lowest number of disapprovals at each timestep (with lexicographical tie-breaking) is strategyproof. We will refer to this algorithm as \textsc{Greedy Min-Sum}.

\begin{theorem} \label{thm:minsum_sp}
    \textsc{Greedy Min-Sum} is strategyproof.
\end{theorem}
A natural follow-up question is whether the \minsum{} objective is compatible with a stronger version of strategyproofness.
A well-known generalization of strategyproofness is \emph{group strategyproofness} (GSP). Intuitively, GSP states that no group of agents should be able to misreport their preferences so as to benefit every member of the group. 
Unfortunately, we show that the \minsum{} objective is incompatible with GSP.

We start by presenting the formal definition of group strategyproofness.
Recall that agent $i$'s disutility function $d_i$ is computed with respect to her (true) disapproval vector $\mathbf{D}_i$.
\begin{definition}[Group-strategyproofness]
    For each $S \subseteq N$, let $\mathcal{D}_{-S}$ denote the list of all disapproval vectors except those of agents in $S$.
    A mechanism $\mathcal{M}$ is \emph{group-strategyproof} (GSP) if for each instance $(N,P,T,(\mathbf{D}_i)_{i \in N})$, each subset of agents $S \subseteq N$, each agent $i \in S$ 
    and each list of disapproval vectors $(\mathbf{D}'_j)_{j \in S}$, it holds that $d_{i}(\mathcal{M}(\mathcal{D}_{-S}, (\mathbf{D}_j)_{j \in S})) \leq d_{i}(\mathcal{M}(\mathcal{D}_{-S},(\mathbf{D}'_j)_{j \in S}))$.
\end{definition}

Note that GSP reduces to SP if we only consider singleton groups.
Then, our negative result is as follows.

\begin{proposition} \label{prop:minsum_gsp}
    Let $\mechanism$ be a mechanism that always returns a \minsum{} outcome. Then $\mechanism$ is not group-strategyproof.
\end{proposition}

Notably, while we only defined GSP (and SP) for deterministic mechanisms, the above negative result also applies to an analogous definition of GSP for \emph{randomized} mechanisms.
This is because the instance constructed in the counterexample has a unique \minsum{} outcome, and any randomized mechanism behaves exactly like a deterministic one.

The above results indicate that the \minsum{} objective is compatible with disincentivizing strategic manipulation by individuals (Theorem~\ref{thm:minsum_sp}), but not by groups (Proposition~\ref{prop:minsum_gsp}).
We further show that, while groups may have opportunities for manipulation, identifying a `suitable manipulation' that strictly benefits every agent within the group can be computationally intractable.
Following the conventions of the literature on voting manipulation \cite{manip-handbook}, we assume the group has knowledge of the reported disapproval vectors of agents outside the group.

\begin{theorem}\label{thm: NP-c, GSP}
    Let $\mathcal{M}$ be a mechanism that always returns a \minsum{} outcome. Given an instance $(N,P,T,(\mathbf{D}_i)_{i \in N})$ and a subset of agents $S \subseteq N$, determining whether there exists disapproval vectors $(\mathbf{D}'_i)_{i \in S}$ such that $d_{i}(\mathcal{M}(\mathcal{D}_{-S},(\mathbf{D}_i)_{i \in S})) > d_{i}(\mathcal{M}(\mathcal{D}_{-S},(\mathbf{D}'_i)_{i \in S}))$ for all $i \in S$ is \emph{NP-complete}.
\end{theorem}

\subsection{Mechanisms for \minmax{}}
Next, we turn to the same questions for the \minmax{} objective. However, in contrast, we show that \minmax{} is fundamentally incompatible with strategyproofness. Intuitively, agents are incentivized to appear `worse off' by misreporting disapproval for projects they do not actually disapprove of.
\begin{proposition} \label{prop:minmax_gsp}
    Let $\mechanism$ be a mechanism that always returns a \minmax{} outcome. Then $\mechanism$ is not strategyproof.
\end{proposition}
Again, while we only defined SP for deterministic mechanisms, the above negative result also applies to an analogous definition of SP for randomized mechanisms.

Despite the stronger negative result above, we show that for any mechanism returning a \minmax{} outcome with ties broken lexicographically, it may be computationally intractable for any agent to find a manipulation that decreases their disutility. Our proof places this problem at the second level of the polynomial hierarchy; this is because computing a \minmax{} outcome is already a hard problem (see Theorem~\ref{thm: NP-hard}).

\begin{theorem} \label{thm_sigma2p}
    Let $\mathcal{M}_\text{lex}$ be a mechanism that returns a \minmax{} outcome with lexicographical tiebreaking.
    Given some instance $(N, P, T, (\mathbf{D}_i)_{i \in N})$ and an agent $i \in N$, determining whether there exists a disapproval vector $\mathbf{D}'_i$ such that $d_i(\mechanism_\text{lex}(\mathcal{D}_{-i},\mathbf{D}_i)) > d_i(\mechanism_\text{lex}(\mathcal{D}_{-i},\mathbf{D}'_i))$ is $\Sigma^P_2$\emph{-complete}.
\end{theorem}

\section{Online Setting} \label{sec:online}
In practice, multi-year plans often face discontinuities due to political shifts or other unforeseen changes, leading to potential failure in execution mid-way. Consequently, ensuring fairness for voters only at the conclusion of X years may be inadequate; voters might instead expect satisfaction at every consecutive timestep. 
Further, project plans (and hence availability/feasibility) themselves may evolve over time, making it difficult, if not impossible, to predict the project availability in advance.
In both of these scenarios, it is natural to consider an \emph{online} setting where agents' preferences for future projects are unknown, and decisions are made without access to future information. We utilize \emph{competitive analysis} for online algorithms to measure the impact of this lack of future information on our ability to achieve optimal welfare.

Let $\mathcal{B}$ be an {\em online algorithm} for our setting, i.e., an algorithm that for each $k\in T$
selects $o_k$ based on $(D_{it})_{i\in N, t\in[k]}$. Let $\mathcal{B}(\mathcal{I})$ denote the output of $\mathcal B$ on instance $\mathcal I$.

\begin{definition}[Competitive Ratio]
    For an online algorithm $\mathcal{B}$ and objective $W \in \{\minsum{}, \minmax{}\}$, the \emph{competitive ratio (CR)} of $\mathcal{B}$ with respect to $W$ is the supremum over all instances $\mathcal{I}$ of the ratio between the $W$-value of $\mathcal {B(I)}$ and the minimum $W$-value of an outcome for $\mathcal{I}$:
    \begin{equation*}
        \text{CR}_W(\mathcal{B}) = \sup_{\mathcal{I}} \frac{ W_\text{op}(\mathcal{B(I)})}{\min_{\mathbf{o} \in \Pi(\mathcal{I})} W_\text{op}(\mathbf{o})}.
    \end{equation*}
\end{definition}

Note that \textsc{Greedy Min-Sum} is an online algorithm, and it outputs a
\minsum{} allocation. Hence,  CR$_\minsum(\mathcal{\textsc{Greedy Min-Sum}}) = 1$.

However, for the \minmax{} objective, the lack of information about the future can significantly impact the \minmax{} value. For the remainder of this section, we focus on the competitive ratio with respect to \minmax{}. 

Similarly to the greedy algorithm for \minsum{}, 
we consider the greedy egalitarian algorithm, which, at each timestep $k$, picks a project so as to minimize the maximum disutility over the first $k$ timesteps.
We refer to this algorithm as \textsc{Greedy Min-Max}. 
Unfortunately, it turns out that for both \textsc{Greedy Min-Max} and \textsc{Greedy Min-Max} the competitive ratio is lower-bounded by $\Omega(n)$.

\begin{proposition} \label{prop:online_greedy_competitive}
    \textsc{Greedy Min-Sum} and \textsc{Greedy Min-Max} both have a competitive ratio of $\Omega(n)$ with respect to \minmax{}.
\end{proposition}
Next, we present a lower bound for \emph{all} online algorithms. 
We consider a weak, non-adaptive adversary that does not have access to the randomized results of the algorithm and show that even against such an adversary, any online algorithm has a competitive ratio of at least $\Omega(\log n)$.

\begin{proposition} \label{prop:comp_ratio_lowerbound}
    Against a non-adaptive adversary, any online algorithm (deterministic or randomized) has a competitive ratio of $\Omega(\log n)$ with respect to $\minmax{}$.
\end{proposition}

\section{Other Fairness Notions} \label{sec:other}
Throughout this paper, we focused on the \minmax{} objective---a fairness criterion that has been extensively studied across many topics in social choice. A natural extension of this work would be to explore other well-established fairness concepts, such as envy-freeness, proportionality or equitability. While envy-based measures are commonly studied in the fair division literature, these do not adapt well to the public goods or chores setting. In the remainder of this section, we discuss proportionality and equitability in more detail.

\subsection{Proportionality}
The first concept we consider is a widely studied notion of fairness in both the fair division of private goods/chores and public goods settings, called \emph{proportionality} \cite{barman2019integralequilibria,branzei2023competitivechores,conitzer2017fairpublic}. 
More specifically, for private chores, proportionality mandates that no agent should receive more than $1/n$ of their pessimal (worst-case) disutility; and for private/public goods, it requires that all agents receive at least $1/n$ of their optimal (best-case) utility. 
While proportionality may not always be achievable in these settings, mild relaxations of the concept are known to always exist. 
However, an analogous definition for public chores proves to be unintuitive, as illustrated by the following example.
\begin{example}
    Consider an instance with $n$ agents, $\ell$ timesteps, and set of projects $P = \{p_1, \dots, p_n\}$. At every timestep, each agent $i \in N$ disapproves all projects in $P \setminus \{p_i\}$. 
    
    Then, every outcome $\mathbf{o}$ is disapproved by $n-1$ agents in each round, so by the pigeonhole principle we have 
    $d_i(\mathbf{o}) \geq \ell\cdot\frac{n-1}{n}$ for some $i\in N$. This is despite the fact that $i$'s disutility from $(p_i, \dots, p_i)$ is $0$. Thus, no outcome is close to being proportional for the standard definition of proportionality. 
\end{example}
Identifying a suitable notion of proportionality for the public chores setting remains an intriguing open problem.

\subsection{Equitability}
Another potentially suitable fairness notion is \emph{equitability}, which mandates that agents' disutilities should be equal \cite{elkind2022temporalslot,freeman2019eq}. 
While it may not always be achievable in practice (this is also the case for many similar fairness properties in the social choice literature), we may be interested in obtaining an equitable outcome when one exists.
Unfortunately, we can show that even determining if an instance admits an equitable outcome is computationally intractable.
\begin{theorem} \label{thm:eq-np-complete}
    Determining if there exists an equitable outcome is \NP-complete.
\end{theorem}
Moreover, enforcing equitability comes at a very high cost to (both \minsum{} and \minmax{}) welfare. Our definition of {\em price of equitability} is structurally similar to the definition of the price of temporal fairness, 
and leads to the following result.
\begin{theorem} \label{thm:price_of_EQ}
    The price of equitability is $\Omega(n^2)$ with respect to \minsum{} and at least $\Omega(n)$ with respect to \minmax{}.
\end{theorem}

Exploring more suitable fairness concepts in this setting remains an intriguing direction for future work.

\section{Conclusion} \label{sec:conclusion}
In this work, we introduced and studied a model of temporal voting where agents can express disapprovals over candidates.
We investigated the computational complexity of two well-studied welfare objectives---\minsum{} and \minmax{}---and identified several settings where the problem can be solved efficiently, together with accompanying algorithms.
We also quantified the effects of enforcing temporal fairness on social welfare, and analyzed the strategic implications associated with these welfare objectives.
Further, we derived bounds on the price of temporal fairness and the competitive ratio of algorithms in the online setting.
Finally, we made a case for why proportionality and equitability may not be suitable as fairness measures for the public chores setting.

Directions for future work include defining and studying weaker forms of strategyproofness that may be compatible with \minmax{} in this setting, or identifying an appropriate (potentially weaker) notion of proportionality in this setting. It would also be interesting to consider a model with both positively- and negatively-valued candidates.

\bibliographystyle{named}
\bibliography{abb,ijcai25}

\onecolumn
\appendix
\begin{center}
\Large
\textbf{Appendix}
\end{center}

\vspace{2mm}

\section{Omitted Proofs in Section \ref{sec:welfare}}

\subsection{Proof of Theorem \ref{thm:gamma=1} (continued)}
Indeed, suppose there is an outcome $\mathbf o$ that satisfies all constraints. We can transform it into a matching as follows. For each $k\in T$, we determine the unique agent $i$ who disapproves $o_k$ in timestep $k$.
    Let $\lambda$ be the disutility of $i$
    from the partial outcome $(o_1, \dots, o_k, \varnothing, \dots, \varnothing)$. We then match $(i, \lambda)$ to $k$.
    We claim that this is a valid matching of size $|T|$. Indeed, $G$ contains the edge from $(i, \lambda)$ to $k$: clearly, $i$ disapproves a project at timestep $k$ (namely, $o_k$), and, moreover, if it was the case that $(k', \lambda')\in{\mathbf A}$ for some $k'\ge k$, $\lambda'<\lambda$, then $\mathbf o$ would violate that constraint. Also, each vertex $(i, \lambda)$ is matched to at most one timestep $k$, namely, the one where the disutility of $i$ changes from $\lambda-1$ to $\lambda$.
    
    Conversely, suppose $G$ admits a matching $M$ of size $|T|$.
    By construction, if our graph contains an edge from $(i, \lambda)$ to $k$, it also contains an edge from $(i, \lambda')$ to $k$ for each $\lambda'<\lambda$.
    Hence, we can assume that if a pair $(i, \lambda)$ is in $M$, then so are all pairs of the form $(i, \lambda')$ with $\lambda'<\lambda$. Now, to construct $\mathbf o$, for each $k\in T$ we identify a pair $(i, \lambda)$ such that $(i, \lambda)$ is matched to $k$, and select $o_k$ from $D_{ik}$ (the presence of the edge $\{(i, \lambda), k\}$ ensures that $D_{ik}\neq\varnothing$). We will argue that $\mathbf o$ satisfies all constraints in $\mathbf A$. Indeed, suppose $\mathbf o$ violates a constraint $(k', \lambda')\in{\mathbf A}$. Then there is an agent $i$ whose disutility at timestep $k'$ is $\lambda>\lambda'$. Let $k\le k'$ be the timestep where $i$'s disutility became $\lambda$. Then there is an edge from $(i, \lambda)$ to $k$ in our graph, but $\lambda'<\lambda$, $k\le k'$, a contradiction with how the edges of $G$ are constructed. 

\subsection{Proof of Theorem \ref{thm: NP-hard}}
We reduce from the NP-complete \textsc{3-Occur-3SAT} problem. \textsc{3-Occur-3SAT} is a restricted instance of SAT where all clauses have a maximum length of $3$ and all variables occurs at most $3$ times. 
An instance of \textsc{3-Occur-3SAT} is defined by a boolean formula $F$ with $n$ variables $x_{1},\dots,x_{n}$ and $m$ clauses $c_{1},\dots,c_{m}$. For each $i\in[n]$, the literal $x_{i}$ appears twice and the literal $\bar{x}_{i}$ appears once. $F$ is a \texttt{YES} instance if and only if there exist an assignment for the variables such that $F$ evaluates to \texttt{TRUE}.

We first perform a polynomial-time preprocessing step that yields an equisatisfiable formula: 
for all clauses with just a single literal, we assign \texttt{TRUE} to that literal.
This gives us a formula $F'$ that is equisatisfiable to $F$.
Also note that all clauses in $F'$ are of length $2$ or $3$, and all literals occurs at most twice in $F'$.
Next, let $X = \{x_1, \dots, x_n \}$ and $\mathcal{C} = \{C_1, \dots C_m\}$  be the set of $n$ variables and the set of $m$ clauses in $F'$, respectively. Furthermore, let $\mathcal{C}_2  = \{C'_1, \dots, C'_{m'} \} \subseteq \mathcal{C}$ be the set of $m'$ clauses of length $2$, where $m' \leq m$.

In constructing our instance of \minmaxdec{}, we will have $m$ agents (each agent corresponding to a clause), and $m'+ n$ timesteps. 
For $i \in [m']$, let timestep $t = n + i$ contain one project that is in the disapproval set of the agent that corresponds to $C'_i$. 
Moreover, for $i \in [n]$, let timestep $t = i$ contain two projects, with $p_1$ corresponding to setting $x_i$ to \texttt{TRUE} and $p_2$ corresponding to setting $x_i$ to \texttt{FALSE}. 
For agent $j \in [m]$, let $D_{jt} = \{p_1\}$ if $\neg x_i \in C_j$, $D_{jt} = \{p_2\}$ if $x_i \in C_j$ and $D_{jt} = \varnothing$ otherwise. 
Intuitively, if the agent corresponding to a clause disapproves a project, then we assign some literal in the clause to \texttt{FALSE}. 
Note that as all literals occur at most twice, for each timestep, at most $2$ agents will disapprove the project. 

Then we will prove that $F'$ is satisfiable if and only if  for $\textbf{A} = \{(m' + n , 2)\}$, i.e., there exists an outcome $\mathbf{o}$ such that for all $i \in [m]$, $d_i(\mathbf{o}^{(m'+n)}) \leq 2$.

For the `if' direction, the outcome corresponding to assigning the variables $(o_{m' + 1}, \dots, o_{m' + n})$ directly translates to a Boolean assignment of the variables in $F'$. For each clause $C_i$ of length $3$, $d_i(\mathbf{o}^{(m'+n)}) - d_i(\mathbf{o}^{(m')}) \leq 2$ implies that at most $2$ literals in the clause $C_i$ is \texttt{FALSE}.  
For each clause $C'_{i'}$ of length $2$ (which corresponds to some agent $i$), because of the project at timestep $i'$ that has to be chosen, we have $d_i(\mathbf{o}^{(m'+n)}) \leq 1$, and at most $1$ literal in the clause $C'_{i'}$ is \texttt{FALSE}. 
Thus, the Boolean assignment corresponding to $\textbf{o}$ will assign \texttt{TRUE} to some literal in each clause, implying that $F'$ is satisfiable.

For the `only if' direction, suppose we have a Boolean assignment satisfying $F'$ that assigns at least one literal in every clause to \texttt{TRUE}. 
Consider the the outcome $\mathbf{o}$ that directly corresponds to such an assignment. 
We have that for each clause $C_i$, $d_i(\mathbf{o}^{(m'+n)}) - d_i(\mathbf{o}^{(m')}) \leq 2$ if the clause is length $3$, and $d_i(\mathbf{o}^{(m'+n)}) - d_i(\mathbf{o}^{(m')}) \leq 1$ if the clause is length $2$. Thus, for all $i \in [m]$, we have that $d_i(\mathbf{o}^{(m'+n)}) \leq 2$.

\subsection{Proof of Theorem \ref{thm: FPT}}
We begin by performing a preprocessing step.
    For each project $p \in P$, create $\ell$ copies of it, i.e., 
    replace $p$ with projects $\{p^1, \dots, p^\ell\}$ in $P$, and modify the disapproval sets as follows:
    for each $i \in N$ and $k \in T$, if $p\in D_{ik}$, then replace $p$ with projects $\{p^1, \dots, p^\ell\}$ in $D_{ik}$; and if $p \notin D_{ik}$, then add all projects $\{p^1, \dots, p^{k-1}, p^{k+1}\dots, p^\ell\}$ to $D_{ik}$.

    In this modified instance, for each project $p \in P$, there is at most one timestep $k \in T$ such that $p\notin\cap_{i\in N}D_{ik}$.
    
    Then, we define the {\em type} of a project as the set of agents who do not disapprove it: 
    let the type of project $p \in P$ be $\sigma(p) = \{i\in N: p\notin D_{ik} \text{ for some $k\in T$}\}$. 
    Recall that from the preprocessing step, we have that for each $p\in P$, there exists a timestep $k\in T$ such that $p\notin D_{ik}$ for all $i\in \sigma(p)$, and $p\in D_{ik'}$ for all $i\in N$, $k'\in T\setminus\{k\}$.
    Then, there are at most $2^n$ different project types. 

    For the set of constraints $\textbf{A} = \{(k_1, \lambda_1), \dots, (k_\tau, \lambda_\tau)\}$, we can partition the timesteps into $\tau + 1$ groups $T_1, \dots, T_{\tau+1}$ where $T_i = \{k_{i-1} + 1, \dots, k_i\}$.
    Then, $k_0$ corresponds to (a dummy) timestep $0$, whereas $k_{\tau + 1}$ corresponds to the last timestep $\ell$. 
    Note that we can freely select any project in timestep group $T_{\tau+1}$ as there are no further constraints imposed on the outcome. 
    Thus, we can restrict our attention to only timestep groups $T_1$ to $T_\tau$.
    Thus, we can classify timesteps by the types of projects present in them and by the timestep group they belong to, giving us at most $2^{2^n} \cdot \tau$ different \emph{timestep types}. Let $Q\subseteq 2^N$ be the set of all project types and  let $\mathcal{R} \subseteq 2^Q \times  [\tau]$ be the set of all timestep types.
    
    Now, we construct our integer linear program (ILP).
    For each $R \in \mathcal{R}$, let $z_{R,k}$ be the number of
    timesteps of type $R$ present in timestep group $T_k$.
    For each $i \in N$, let $Q_i$ be the set of project types that agent $i$ disapproves of. 
    For each $R \in \mathcal{R}$, $k \in [\tau]$ and $r\in Q$, we introduce an integer variable $x_{R, k,  r}$ representing the number of timesteps of type $R$ in which a project of type $r$ was chosen in timestep group $T_k$.

    The ILP (feasibility problem) is defined as follows:
	\begin{enumerate}[(1)]
		\item $\sum_{r \in R} x_{R, k, r} \geq z_{R,k} \text{ for each } R \in \mathcal{R}$ and $k \in [\tau]$,
		\item $\sum_{R \in \mathcal{R}}\sum_{j \in [k]}\sum_{r \in Q_i} x_{R, j, r} \leq \lambda_k \text{ for each } i \in N$ and $k \in [\tau]$, and
		\item $x_{R, k, r} \geq 0$ for each $R \in \mathcal{R}$, $k\in [\tau]$ and $r\in Q$.
	\end{enumerate}      
    The first constraint mandates that we select at least one project per timestep, whereas the second constraint ensures that all the constraints in $\textbf{A}$ is satisfied.
    
    Then, since there are at most $\mathcal{O}(2^{n + 2^n} \cdot \tau)$ variables in the ILP, the classic result of Lenstra~\shortcite{lenstra1983integer} implies that our problem is FPT with respect to $n + \tau$.

\subsection{Proof of Theorem \ref{thm: XP}}
We describe a dynamic programming-based algorithm that runs in $\mathcal{O}(\ell^{n+1} \cdot m)$ time. 

    Note that for purposes of the dynamic programming algorithm, let there be a dummy timestep $0$.
    For each agent $i\in N$, timestep $k \in T$, and a partial outcome vector $\mathbf{o}^{(k)}$ (i.e., first $k$ entries---excluding timestep $0$---of the outcome vector $\mathbf{o}$ each has a project chosen, whereas the remaining entries are empty), let $c_{ik}(\mathbf{o}^{(k)})  = \sum_{j = 0}^k d_i(o_j)$.
    Then for any partial outcome $\mathbf{o}^{(k)}$, we define the disutility profile up to timestep $k$ to be $C_k(\mathbf{o}^{(k)}) := (c_{1k}(\mathbf{o}^{(k)}), \dots , c_{nk}(\mathbf{o}^{(k)})) \in [\ell]^n$.

    For each $u' \in [\ell]^n$ and timestep $k \in T$, let $Q_{u',k} = 1$ if and only if there exists some partial outcome $\mathbf{o}^{(k)}$ such that $C_k(\mathbf{o}^{(k)}) = u'$ and any constraint $(k_i,\lambda_i)$ is satisfied if $\lambda_i \leq k$; and $Q_{u',k} = 0$ otherwise.
    For the base case, we have that $Q_{u,0} = 1$ if $u = (0, \dots, 0)$ and $Q_{u,0} = 0$ otherwise.
    
    Then, let $\mathbf{I}_{i,p,k} $ be an indicator function that takes on the value $1$ if and only if agent $i \in N$ disapproves of project $p \in P$ at timestep $k \in T$. 
    Also let $f(u',p,k) = (u'_i - \mathbf{I}_{i,p,k} \mid i \in N)$, where $u'_i$ is the $i$-th component of $u'$.

    For each $k' \in T$, if there is no constraint $(k', \lambda) \in \textbf{A}$ for some $\lambda$, let 
    \begin{equation*}
        Q_{u,k'} =
        \begin{cases}
          0 & \text{if  }  \forall p \in P \text{ such that } Q_{f(u,p,k'),k'-1} = 0 \\
            1 & \text{otherwise} \\
        \end{cases}
    \end{equation*}
    
    Otherwise, if there is the constraint $(k', \lambda) \in \textbf{A}$ for some $\lambda$, then let
    \begin{equation*}
        Q_{u,k'} =
        \begin{cases}
            0 & \text{if } \max(u) \leq \lambda \\
            0 & \text{if }  \forall p \in P \text{ such that } Q_{f(u,p,k'),k'-1} = 0\\
            1 & \text{otherwise} \\
        \end{cases}
    \end{equation*}
    At the end, return `yes' if $Q(u,\ell) = 1$ for some $u \in [\ell^n].$

    Since there are $l^{n+1}$ possible $(u,k)$ pairs, and it takes $\mathcal{O}(m)$ time to update $Q_{u,k}$, this dynamic programming algorithm runs in time $\mathcal{O}(\ell^{n+1} m)$.
    
\subsection{Proof of Theorem \ref{thm: W[2]}}
To show that \minmaxdec{} is XP with respect to $\ell$, observe that there are $m^\ell$ possible outcomes. Thus, when $\ell$ is constant, the total number of outcomes is polynomially-bounded. We can thus iterate through all the outcomes and return `yes' if and only if there is an outcome that satisfies $\textbf{A}$.
	
To show W[2]-hardness, we will show that \minmaxdec{} is W[2]-hard even if $\tau = 1$.
    We reduce from the \textsc{Dominating Set (DS)} problem.
An instance of \textsc{DS} consists of a graph $G = (V,E)$ and an integer $\kappa$; it is a yes-instance if there exists a subset $D \subseteq V$ such that $|D| \leq \kappa$ and every vertex of $G$ is either in $D$ or has a neighbor in $D$, and a no-instance otherwise.
\textsc{DS} is known to be W[2]-complete with respect to the parameter $\kappa$ \cite{niedermeier_invitation_2006}.

Given an instance $(G, \kappa)$ of \textsc{Dominating Set} with $G = (V,E)$, $V=\{v_1, \dots, v_n\}$,  
set $N=[n]$, $P = \{p_1,\dots,p_n\}$, $\ell = \kappa$.
Then for each $i\in N$
and $k\in T$
let $D_{ik}=P \setminus \{p_j : i=j\text{ or }\{v_i, v_j\}\in E\}$.
We claim that $G$ admits a dominating set $D$ with $|D| \leq \kappa$ if and only if there exists an outcome $\mathbf{o}$ such that $d_i(\mathbf{o}) \leq \kappa - 1$ for all agents $i \in N$. 

For the `if' direction, consider an outcome $\mathbf{o} = (p_{j_1},\dots,p_{j_\kappa})$
such that each agent does not disapprove a project in at least $1$ of the timesteps, and 
set $D=\{v_{j_1}, \dots, v_{j_\kappa}\}$. Then $D$ is a dominating set of size at most $\kappa$. 
Indeed, consider any vertex $v_i\in V$. Since agent $i$ does not disapprove $p_{j_k}$ for some $k\in T$, we have that $v_{j_k}\in D$ 
and either $i=j_k$ or $\{v_i, v_{j_k}\}\in E$.
Note that if there are projects chosen more than once, we can simply let $|D|<\kappa$. 

For the `only if' direction, observe that a dominating set $D=\{v_{j_1}, \dots, v_{j_s}\}$ with $s \leq \kappa$ can be mapped to an outcome $\mathbf{o} = (p_{j_1}, \dots, p_{j_s}, p_1, \dots, p_1)$, where $p_1$ is selected at the last $\kappa-s$ timesteps. As any vertex in $G$ is either already in $D$ or has some neighbor in $D$, we have that $d_i(\mathbf{o}) \leq \kappa - 1$ for each agent $i \in N$.

\section{Omitted Proofs in Section \ref{sec:potf}}
\subsection{Proof of Theorem \ref{thm:potf_minmax}}
    For constraints $\textbf{A} = \{(k_1, \lambda_1), \dots, (k_\tau, \lambda_\tau)\}$, we can partition the set of timesteps into $\tau + 1$ groups $T_1, \dots, T_{\tau+1}$ where $T_i = \{k_{i-1} + 1, \dots, k_i\}$ with $k_0$ corresponding to timestep $0$ and $k_{\tau + 1}$ corresponding to timestep $\ell$.
    Let timestep $0$ be a dummy timestep with no disapproved projects.
    
    We first prove the lower bound.  Let $\tau' = \min(\tau, n)$ and $P = \{p_1,p_2\}$.
    Then, for each $i \in [\tau']$, let timestep group $T_i$ consist of two identical timesteps where agent $i$ disapprove only $p_1$ and all other agents disapprove only $p_2$. 
    For the constraint $(k_i,\lambda_i)$, we have that $\lambda_i = i$. Let timestep group $T_{\tau'+1}$ be empty.

    Then, the outcome that selects $p_1$ at all timesteps has \minmax{} value of $2$. 
    However, the constraints in $\mathbf{A}$ forces us to select $p_1$ and $p_2$ at one timestep each in every timestep group. Thus, the \minmax{} value of any outcome that satisfy the temporal constraints is $\tau'$ and the $\text{PoTF}_\minmax = \frac{\tau'}{2} = \Omega(\tau') = \Omega(\min(\tau,n))$.
    
    For our upper bound, we will show $\text{PoTF}_\minmax =\mathcal{O}(\tau)$ and then show $\text{PoTF}_\minmax =\mathcal{O}(n)$.

    Let $\mathbf{o}^*$ be the optimal outcome and let $W^*$ be the \minmax{} value of $\mathbf{o}^*$. Then if there is an outcome $\mathbf{o}$ satisfying all the constraints in $\textbf{A}$, then there is an outcome $\mathbf{o}'$ satisfying all the constraints in $\textbf{A}$ such that $\minmax{}(o') \leq (\tau + 1) W^*$ and hence the PoTF$_\minmax$ is at most $\mathcal{O}(\tau)$.

    Let $i \in \{0\} \cup [\tau]$ be the largest integer such that $\lambda_i \leq i W^*$ where $(k_0,\lambda_0) = (0,0)$. 
    We now add the constraint $(\ell, \lambda_i +   W^* )$. 
    First note that such a constraint is stricter than all constraints $(k_j, \lambda_j$) for $j > i$ where if an outcome satisfies all constraints $(k_j, \lambda_j)$ when $j\leq i$ and the constraint $(\ell, \lambda_i +   W^*)$, then the outcome satisfies $\textbf{A}$. 
    Furthermore, as $\lambda_i \leq \tau W^*$, the \minmax{} value of the outcome is at most $(\tau + 1) W^*$.

    If there exists at outcome $\mathbf{o}$ that satisfies the constraints in $\textbf{A}$, then construct $\mathbf{o}'$ by selecting the same projects as $\mathbf{o}$ for the first $k_i$ timesteps, and the same projects as $\mathbf{o}^*$ for the remaining timesteps. 
    Such an outcome will satisfy all constraints $(k_j, \lambda_j)$ when $j\leq i$, by our assumption. 
    Moreover, for all agents $x$, we have that $\sum_{j=1}^{\ell} d_x(o_j) = \sum_{j=1}^{k_i} d_x(o_j) + \sum_{j=k_i + 1}^{\ell} d_x(o_j) \leq \lambda_i + W^*$.

    At any timestep, if there is a project that is disapproved by no one, then it should always be selected. 
    Thus, without loss of generality, we can assume that at every timestep, every project is disapproved by some agent. 
    Then, any outcome has a \minsum{} value of at least $\ell$, and so every result has an \minmax{} value of at least $\ell / n$. Furthermore, the \minmax{} value is trivially at most $\ell$. Thus, $\text{PoTF}_\minmax \leq \frac{\ell}{\ell / n} = n$.

\section{Omitted Proofs in Section \ref{sec:sp}}
\subsection{Proof of Theorem \ref{thm:minsum_sp}}
    \textsc{Greedy Min-Sum} picks the project that has the least number of disapprovals at each timestep (with lexicographical tiebreaking, as mentioned earlier).

    Consider some agent $i \in N$.
    Any project that $i$ does not disapprove of, he is said to \emph{approve}.
    Let agent $i$'s disapproval vector under truthful preferences be $\mathbf{D}_i = (D_{i1},\dots,D_{i\ell})$. 
    Note that agent $i$ can only lie by misreporting one or more of $D_{ik}$ (for some $k \in T$) either by including projects she approves of and/or by excluding projects she disapproves of.
    For each $k \in T$, when selecting $o_k$, \textsc{Greedy Min-Sum} only considers the information in $D_{ik}$ and disregards the rest of $\mathbf{D}_i$. Hence, misreporting $D_{ik}$ can only affect the outcome of $o_k$ and only when $o_k$ changes from a project in $D_{ik}$ to a project not in $D_{ik}$ can misreporting $D_{ik}$ decrease the disutility of agent $i$.

    We will now show that no matter the preferences of other agents in $N\setminus \{i\}$ at timestep $k$, misreporting $D_{ik}$ can never change $o_k$ from a project in $D_{ik}$ to a project not in $D_{ik}$. Hence, agent $i$ cannot obtain a lower disutility by misreporting her disapproval set.
    
    Suppose by truthfully reporting $D_{ik}$, $o_k=c$ and $c \in D_{ik}$. That means for each $c' \notin D_{ik}$, either the number of disapprovals for $c'$ is more than the disapprovals for $c$, or both projects have an equal number of disapprovals and $c$ is preferred to $c'$ in the tie-breaking order. 
    Both the actions of disapproving projects agent $i$ approves of and excluding projects agent $i$ disapproves of cannot decrease the number of disapprovals any projects not in $D_{ik}$ or increase the number of disapprovals of $c$. 
    Hence, no matter how agent $i$ misreports her preference at timestep $k$, no project she approves of will be chosen.

\subsection{Proof of Theorem \ref{prop:minsum_gsp}}
    Consider the instance with five agents $N = \{1,2,3,4,5\}$, two projects $P = \{p_1,p_2\}$, and three timesteps.
    For any agent $i \in \{1,2,3\}$, let agent $i$ disapprove $p_1$ only at timestep $i$, and let agents $4$ and $5$ disapprove $p_2$ at all three timesteps. Then, the only \minsum{} outcome is $\mathbf{o} =(p_1,p_1,p_1)$. For each agent $i\in\{1,2,3\},$ we have that $d_i(\mathbf{o}) = 1$. 
    However, if all agents in $\{1,2,3\}$ misreport that they disapprove $p_1$ in all timesteps, then the only \minsum{} outcome is $\mathbf{o}' = (p_2,p_2,p_2)$ and $d_i(\mathbf{o}) = 0$ for all $i \in \{1,2,3\}$, thereby violating group-strategyproofness.
\subsection{Proof of Theorem \ref{thm: NP-c, GSP}}
We once again reduce from the NP-complete \textsc{3-Occur-3SAT} problem. \textsc{3-Occur-3SAT} is a restricted instance of SAT where all clauses have a maximum length of $3$ and all variables occurs at most $3$ times. 
Consider a Boolean formula $F$. 
Then, $F$ is satisfiable if there exists some assignment of Boolean values to variables such that the conjunction of all clauses evaluate to \texttt{TRUE}; and is not satisfiable otherwise.

We first perform a polynomial-time preprocessing step that yields an equisatisfiable formula: 
for all clauses with just a single literal, we assign \texttt{TRUE} to that literal.
This gives us a formula $F'$ that is equisatisfiable to $F$.
Also note that all clauses in $F'$ are of length $2$ or $3$, and all literals occurs at most twice in $F'$.
Next, let $X = \{x_1, \dots, x_n \}$ and $\mathcal{C} = \{C_1, \dots C_m\}$  be the set of $n$ variables and the set of $m$ clauses in $F'$, respectively. Furthermore, let $\mathcal{C}_2  = \{C'_1, \dots, C'_{m'} \} \subseteq \mathcal{C}$ be the set of $m'$ clauses of length $2$, where $m' \leq m$. Furthermore, if $m \leq 5$, then we can check the satisfiability of $F$ in constant time. Thus, without loss of generality, let $m > 5$.

For our instance, we have $n + 1$ timesteps and $m + 4$ agents (corresponding to the clauses $\mathcal{C}$ and a set of $3$ dummy agents $N_D$). 
First, we have a timestep with $1$ project that all agents corresponding to clauses in $\mathcal{C}_2$ disapproves. 
Then, the rest of the timesteps correspond to assigning a variable $x_i$ to \texttt{TRUE} or \texttt{FALSE} and there are $3$ projects. Agents corresponding to a clause $C$ with the literal $x_i \in C$ disapproves of $p_1, p_3$, and agents corresponding to a clause $C$ with the literal $\neg x_i \in C$ disapproves of $p_2,p_3$ and agents in $N_D$ all disapprove $p_1, p_2$. 

Then, $F$ is satisfiable if the agents corresponding to $C$ can misreport their preference such that all agents corresponding to $C$ can decrease their disutility. 

We note that for all the timesteps except the first, $p_1$ and $p_2$ are each disapproved by $4$ or $5$ agents and $p_3$ is disapproved by $3$ agents. 
Hence, $p_3$ is always chosen and all agents corresponding to a clause in $C$ have a disutilty of at most $3$. 
However, as $m > 5$, the agents in $C$ can ensure $p_1$ or $p_2$ is selected again. 
Furthermore, as the agents that disapprove $p_3$ in these timesteps is a strict superset of the agents (corresponding to an agent in $C$) that disapprove $p_1$ or $p_2$, they would always rather $p_1$ or $p_2$ be selected instead. 

Thus, the problem becomes: is an outcome consisting of $p_1,p_2$ such that all agents corresponding to $C$ disapproves of projects chosen in at most $2$ timesteps. 
By a similar reasoning as that in Theorem \ref{thm: NP-hard}, this problem is NP-hard.

\subsection{Proof of Proposition \ref{prop:minmax_gsp}}
Consider an instance with three agents $N = \{1,2,3\}$, three projects $P = \{p_1,p_2,p_3\}$, and two timesteps.
    Let agents' disapproval sets $\mathbf{D}_1, \mathbf{D}_2, \mathbf{D}_3$ be such that $D_{11} = D_{21} = D_{31} = \{p_2,p_3\}$, and $D_{i2} = \{p_i\}$ for all $i \in \{1,2,3\}$.

    If $p_1$ is not selected at the first timestep, there is an agent that disapproves the project chosen at both timesteps, so the \minmax{} value is $2$. Thus, for every \minmax{} outcome $\mathbf{o} = (o_1,o_2)$ we have $o_1 = p_1$ and $o_2 \in \{p_1,p_2,p_3\}$. This ensures that the \minmax{} value is $1$. 
 
    Assume without loss of generality that $\mechanism$ selects $\mathbf{o} = (p_1, p_2)$ when the agents report truthfully. Then $d_1(\mathbf{o}) = 1$. 
    Then, agent 1 can misreport their disapproval vector as $\mathbf{D}'_1 = (D'_{11}, D'_{12})$, where $D'_{11} = D'_{12}=\{p_1\}$.
    In this case, the only \minmax{} outcome is ${\mathbf o}' = (p_1,p_1)$, so $\mechanism$ is forced to output ${\mathbf o}'$. Moreover, agent~$1$'s disutility (with respect to his true preference) from ${\mathbf o}'$ is $d_1(\mathbf{o}') = 0 < d_1(\mathbf{o})$. 
    Thus, agent $1$ has an incentive to misreport.
    
\subsection{Proof of Theorem \ref{thm_sigma2p}}
Note that since it is NP-complete to even find a \minmax{} outcome (Theorem~\ref{thm: NP-hard}), under $\mathcal{M}_{lex}$, an agent cannot verify in polynomial-time whether a disapproval vector $D'_i$ would give them a strictly lower disutility, thus showing containment in $\Sigma_2^P$.

Next, we prove $\Sigma_2^P$-hardness. 
We reduce from the \textsc{2-Alternating-Quantified-Satisfiability problem}, which was shown by Johannes~\shortcite{phdthesis} to be $\Sigma^P_2$-complete. It is formally defined as follows.
\begin{tcolorbox}[title=\textsc{2-Alternating-Quantified-Satisfiability ($B^{\textit{CNF}}_2$)}]
    \textbf{Input}: Sets $X$ and $Y$ of variables, Boolean formula $F$ over $X$ and $Y$ in conjunctive normal form with exactly 3 variables in each clause.
    \tcblower
    \textbf{Question}: Is there an assignment for $X$ so that there is no assignment for $Y$ such that $F$ is satisfied?
\end{tcolorbox}
 First, if a clause contains both the literal $l$ and $\neg l$, remove the clause. If $|Y| \leq 1$, we can resolve away the variables in $Y$ and accept if there is a clause remaining. Thus, we can focus on the case when $|Y| \geq 1$.

Let $F$ have m clauses $C = \{c_1, \dots, c_m\}$ containing variables in $X$ and $Y$. Note that for all $i\in [m]$, $|c_i| = 3$.

Without loss of generality, assume that if at a timestep both $p_1$ and $p_2$ have the same set of agents disapproving it, then $p_1$ is chosen by $\mathcal{M}_{lex}$ over $p_2$. 

For our instance, let there be $m + 2$ agents, with the first $m$ corresponding to the $m$ clauses in $F$. 
Also, let agent $m + 1$ have the highest priority. 
We have five groups of timesteps $T_1,T_2,T_3,T_X, T_Y$, and let
\begin{itemize}
    \item $T_1$ have $|Y| - 2$ identical timesteps where in each timestep, agent $m + 1$ disapproves $p_1$ while agents $1$ to $m$ disapprove $p_2$;
    \item $T_2$ have two identical timestep, where in each timestep, agent $m+1$ disapproves $p_1$, and agent $m + 2$ disapproves $p_2$;
    \item $T_3$ have only one timestep, where agent $1$ disapproves $p_1$, no agent disapprove $p_2$ and agents 1 to $m$ disapprove $p_3$;
    \item $T_X$ have $|X|$ timesteps where each timestep corresponds to assigning a variable in $X$. 
    For a timestep corresponding to assigning a particular variable $x \in X$, let agent $m+1$ disapproves $p_1$, and agent $i \in [m]$ disapproves $p_2$ if the literal $x \in c_i$ and disapproves $p_3$ if the literal $\neg x \in c_i$; and
    \item $T_Y$ have $|Y|$ timesteps where each timestep corresponds to assigning a variable in $Y$. 
    For a timestep corresponding to assigning a particular variable $y \in Y$, let agent $m+1$ and agent $m+2$ disapprove both $p_1$ and $p_2$, and agent $i \in [m]$ disapproves $p_1$ if the literal $x \in c_i$ and disapproves $p_2$ if the literal $\neg x \in c_i$.
\end{itemize}
Then, we will prove that agent $m+1$ can decrease their disutility by misreporting their preference if and only if there is an assignment $\mathcal{A}$ for $X$ (so that there is no assignment for $Y$) such that $F$ is satisfied.

Observing agent $m+1$ and $m+2$ preference over the projects for timesteps in the groups $T_2, T_Y$, the minimum possible \minmax{} value is $|Y| + 1$. 
This is achievable by selecting $p_2$ for all the timesteps in $T_1$, selecting each of $p_1$ and $p_2$ once in $T_2$, selecting $p_2$ for the timestep in $T_3$, and selecting any project for all timesteps in $T_X,T_Y$. 
For agents $1$ to $m$, as they disapprove exactly 3 project across $T_X,T_Y$, their disutility from this selection is at most $|Y| + 1$.  Furthermore, if the \minmax{} value of the outcome is $|Y| + 1$, then agent $m+1$ will have a disutility of $|Y| + 1$.

For the `if' direction, let $\mathcal{A}$ be an assignment over the variables in $X$ (so that there is no assignment for $Y$) such that $F$ is satisfied. 
Agent $m+1$ should augment his preference for the timesteps in $T_3,T_X$ in the following way:
\begin{itemize}
    \item report that he disapprove $p_2$ for the timestep at $T_3$; 
    \item for a timestep in $T_X$ corresponding to the variable $x$, additionally report that he disapprove $p_3$ if $\mathcal{A}[x] = 1$ and report that he disapprove $p_2$ if $\mathcal{A}[x] = 0$
\end{itemize}
By doing so, agent $m+1$ can guarantee that his disutility is at most $|Y|$.

Under this misreported preference, we show that there is no outcome with a \minmax{} value of at most $|Y| + 1$. 
Suppose that such an outcome exists, let it be $\mathbf{o}$. 
Then $p_2$ should be selected across all timesteps in $T_1$ and $p_3$ must be selected at the timestep in $T_3$. 
Furthermore, for each variable $x \in T_x$ and at the corresponding timestep, if agent $m+1$ disapprove $p_3$, then $p_2$ must selected at that timestep in $\mathbf{o}$, and if  agent $m+1$ disapprove $p_2$, then $p_3$ must selected at that timestep in $\mathbf{o}$. 

We are essentially choosing an outcome for the timestep that corresponds to the assignment under $\mathcal{A}$.
For the timesteps corresponding to $T_1, T_3$, agents $1$ to $m$ have a disutility of $|Y| - 1$. 
Thus, if $\mathbf{o}$ has \minmax{} value of at most $|Y| + 1$, then agents $1$ to $m$ will disapproves at most $2$ projects in $T_X,T_Y$, even when the projects in $T_X$ are selected according to $\mathcal{A}$. 
This corresponds to an assignment $\mathcal{B}$ over $Y$ that satisfy $F$, but which by the premise of $\mathcal{A}$ does not exists.

Hence, all outcomes has \minmax{} value of at least $|Y| + 2$. 
Furthermore, selecting $p_2$ at all timesteps $T_1,T_2$, selecting $p_3$ at $T_3$, selecting the project agent $m+1$ does not disapprove at $T_X$, and selecting any project for timesteps in $T_Y$ yields an outcome with \minmax{} value of at $|Y| + 2$, with agent $m+1$ only disapproving $|Y|$ projects across all timesteps. 
Thus, with lexicographical tie-breaking and $m+1$ having the highest priority, agent $m+1$ is able to decrease their disutility by misreporting their preferences. 

For the `only if' direction, suppose that agent~$m+1$ is able to decrease their disutility by misreporting their preference. Let $\textbf{Z}$ be the misreported preference of agent~$m+1$. 
For all timesteps in $T_1, T_2, T_3, T_X$, agent $m+1$ cannot benefit from disapproving any additional projects (not already disapproved) in those timesteps groups. 
Otherwise, as the set of agents that disapprove any other project available at the timestep is a superset of $\{m+1\}$, $p_1$ will be chosen at that timestep. 
This means that agent $m+1$ will have a disutility of at least $|Y|+1$. 
Furthermore, if agent $m+1$ does not disapprove $p_2$ at the timestep at $T_3$, then there is an outcome that has \minmax{} value of $|Y| + 1$, and agent $m+1$ cannot decrease his disutility. 
Thus, in order for agent $m+1$ to be able to obtain a lower disutility by reporting $Z$, he must disapprove $p_2$ at timestep $2$, with an additional $0$ or $1$ project to disapprove at each timestep of $T_X$.

For $x \in X$ and the timestep in $T_X$ corresponding to $x$, let $\mathcal{A}[x] = 1$ if $p_3$ was disapproved at the timestep and $\mathcal{A}[x] = 0$ otherwise. 
Suppose for a contradiction there is an assignment $\mathcal{B}$ over $Y$ that satisfies $F$ after assigning the variables in $Y$ based on $\mathcal{A}$. 
By selecting projects in $T_X,T_Y$ according to $\mathcal{A} and \mathcal{B}$, agents $1$ to $m$ will only disapprove at most $2$ project across $T_X,T_Y$, and agent $m+1$ does not disapprove any project in $T_X$. 
Additionally, by selecting $p_2$ for all the timesteps in $T_1$, selecting both $p_1$ and $p_2$ once in $T_2$, and selecting $p_3$ for the timestep in $T_3$, all agents will get a disutility of at most $|Y| + 1$.
Thus, as the minimum \minmax{} value is $|Y| + 1$, agent $m+1$ cannot decrease his disutility, a contradiction.

\section{Omitted Proofs in Section \ref{sec:online}}
\subsection{Proof of Proposition \ref{prop:online_greedy_competitive}}
For \textsc{Greedy Min-Sum}, consider the instance with $n = 2k + 1$ agents, $P = \{p_1,p_2\}$, and $k$ timesteps. We note that $k = \Omega(n)$. 
    For each timestep $i \in \{1,\dots,k\}$, let agent $2k+1$ disapprove only $p_1$ and let agents $2i -1$ and $2i$ disapprove only $p_2$. 
    Then, \textsc{Greedy Min-Sum} will return the outcome $\mathbf{o} = (p_1,\dots,p_1)$ with \minmax{} value of $k$.
    However, the outcome $\mathbf{o}'$ has \minmax{} value of $1$. Thus, CR$_\minmax(\textsc{Greedy Min-Sum}) \geq k =\Omega(n)$.

    For \textsc{Greedy Min-Max}, consider the instance with $n$ agents, $P = \{p_1,p_2\}$, and $2n$ timesteps. 
    For each $i \in [n]$, at timestep $2i -1$ and $2i$, let agent $i$ disapprove only $p_1$, and all other agents $i \in [n] \setminus \{i\}$ disapprove only $p_2$.
    Then, for all $i \in [n]$, \textsc{Greedy Min-Max} will return the outcome $\mathbf{o} = (p_1,p_2,p_1,p_2,\dots)$ with \minmax{} value $n$.
    However, outcome $\mathbf{o} = (p_1,\dots,p_1)$ has a \minmax{} value of 2. Thus, CR$_\minmax(\textsc{Greedy Min-Max}) \geq n/2 = \Omega(n)$.
    
\subsection{Proof of Proposition \ref{prop:comp_ratio_lowerbound}}
    Let $n = 2^k$.    
    Now, we construct an instance with two projects and $k$ timesteps. 
    Consider the function $\texttt{RandomSplit} : 2^{[n]}\rightarrow 2^{[n]} \times 2^{[n]}$ that randomly split a set into two equal halves, and the following sequence: $L_0 = [n]$ and $L_i, R_i = \texttt{RandomSplit}(L_{i-1})$ for $i \in [k]$. 
    Let timestep $i$ be chosen with equal probability from one of two possibilities:
    \begin{itemize}
        \item possibility $1$ has agents in $L_i$ disapproving $p_1$ and agents in $R_i$ disapproving $p_2$; and
        \item possibility $2$ has agents in $L_i$ disapproving $p_2$ and agents in $R_i$ disapproving $p_1$.
    \end{itemize}
    Then, note that for any online algorithm, when selecting a project for timestep $i$, the algorithm has no information regarding which project corresponds to $L_i$ and $R_i$.
    
    We also note that by construction, $\cap_{i \in [k]} L_i \neq \varnothing$.
    Thus, the \minmax{} welfare of an outcome $\textbf{o}$ is at least the number of timesteps $i$ such that the project that agents in $L_i$ disapprove was chosen. 
    Let $\mathbf{I}_{L_i}$ be the indicator function that takes the value $1$ if and only if the project that agents in $L_i$ disapprove was chosen. 
    Then the  expected \minmax{} value is $\mathbb{E}[W(\textbf{o})] \geq \sum_{i = 0}^k \mathbb{E}[\mathbf{I}_{L_i}] = k/2$.

    However, if we select all the projects that agents in $R_i$ disapproves for all $i \in [k]$, the \minmax{} welfare is $1$. Thus,  the competitive ratio for any online algorithm with respect to the \minmax{} objective is at least $k/2 = \Omega(\log n)$.

\section{Omitted Proofs in Section \ref{sec:other}}
\subsection{Proof of Theorem \ref{thm:eq-np-complete}}
We reduce from the NP-complete problem \textsc{1-in-3-SAT} \cite{schaefer1978sat}. Given a Boolean formula, $F$, in conjunctive normal form with three literals per clause, the problem is to determine if there is an assignment $\mathcal{A}$ such that each clause has exactly one literal set to \texttt{TRUE}. Let $F$ contain $m$ clauses $\{ c_1 \dots c_m\}$. For a clause $c = (l_1 \lor l_2 \lor l_3)$, let negation of $c_i$ be $neg(c_i) = (\neg l_1 \lor \neg l_2 \lor \neg l_3)$.  Then let $F' = \bigwedge_{i \in [m]} neg(c_i)$ 

For our reduction, introduce an agent for each $c \in F$ and an agent for each $neg(c_i) \in F'$. Introduce a timestep for each variable $x$ present in $F$ containing two projects $p_1$ (which corresponds to assigning $x$ to \texttt{TRUE}) and $p_2$ (which corresponds to assigning $x$ to \texttt{FALSE}). Then, for each clause $c_i$ in $F \land F'$, the agent corresponding to $c_i$ disapproves $p_1$ if $\neg x \in c_i$, disapproves $p_2$ if $x \in c_i$ and disapproves neither projects otherwise.
Finally, we introduce a timestep with one project that all agents corresponding to clauses $neg(c_i) \in F'$ disapprove. Then, there is an equitable outcome if and only if there is an assignment such that each clause in $F$ has exactly one literal set to \texttt{TRUE}.

We note that there is an outcome $\mathbf{o}$ for our instance that corresponds naturally to an assignment $\mathcal{A}$ and vice-versa. Suppose the outcome $\mathbf{o}$ and the assignment $\mathcal{A}$ represents the same assignment. For all $i \in [m]$, if $\mathcal{A}$ sets $k$ literal in $c_i$ to \texttt{TRUE}, then the agent corresponding to $c_i$ has a disutility of $3 - k$ from $\mathbf{o}$ and
the agent corresponding to $neg(c_i)$ has a disutility of $k + 1$ from $\mathbf{o}$.

Suppose there is an assignment $\mathcal{A}$ such that each clause in $F$ has exactly one literal set to \texttt{TRUE}. Let $\mathbf{o}$ be the outcome corresponding to $\mathcal{A}$.  Then, for all agents corresponding to clause $c_i \in F$, the agent would have a disutility of $3 -1 = 2$ from the outcome $\mathbf{o}$. For all agents corresponding to clause $\neg(c_i) \in F$, the agent would have a disutility of $1 + 1 = 2$ from the outcome $\mathbf{o}$. Thus, $\mathbf{o}$ is equitable.

Suppose there is an equitable outcome $\mathbf{o}$. Let $\mathcal{A}$ be the assignment that correspond to the outcome. For clause $c_i$, suppose $\mathcal{A}$ sets $k$ literal in $c_i$ to \texttt{TRUE}. Then, as $\mathbf{o}$ is equitable, $3 - k = k+1$. Thus, as $k = 1$ for all $i \in [m]$, there is an assignment $\mathcal{A}$ such that each clause in $F$ has exactly one literal set to \texttt{TRUE}.

 \subsection{Proof of Theorem \ref{thm:price_of_EQ}}
 Consider the following instance with $n$ agents and $2n-3$ timestep. Timestep $1$ contains $1$ project that only agent $1$ disapproves. At all other timesteps, there are $2$ projects and $p_1$ is disapproved by no agents accross these timesteps. Then, clearly the optimal optimal outcome has \minsum{}/\minmax{} welfare of $1$.

     For timestep $i \in \{2, \dots n\}$, $p_2$ is disapproved by agents $\{i \mod n, \dots, (i + n -3)\mod n \}$. For timestep  $i \in \{n+1, \dots 2n-3\}$, $p_2$ is disapproved solely by agent 1. Selected $p_2$ accross all these timesteps yields an equitable outcome such that for all agents they have a disutility of $n-2$. Furthermore, this is the only equitable outcome. 

     If $p_2$ is not selected across timesteps $\{2, \dots n\}$, then agent 1 will have a greater disutility than all other agents and the outcome is not equitable. Thus, $p_2$ is selected at least once  in timesteps $\{2, \dots n\}$. Suppose $p_2$ was selected $k$ times in $\{2, \dots n\}$ such that $1 \leq k \leq n -1$. The total disutility by agents $\{2, \dots n\} = k(n-2)$ and the only way to outcome is equitable if $k(n-2)$ is divisible by $n-1$. Thus, $k = n-1$ and it follows that $p_2$ must also be selected at $\{n+1, \dots 2n-3\}$.

     Hence, all equitable outcomes have UTIL welfare of at least $(n-2)(n-1)$ and \minmax{} welfare of at least $(n-2)$. The price of equitability is at least $\Omega(n^2)$ with respect to UTIL and at least $\Omega(n)$ with respect to \minmax{}.
\end{document}